\documentclass[12pt,draftcls,onecolumn]{IEEEtran}
\newtheorem{defin}{Definition}
\newtheorem{assump}{Assumption}

\newtheorem{prop}{Proposition}
\newtheorem{alg}{Algorithm}
\newtheorem{remark}{Remark}
\newtheorem{example}{Example}

\newcommand{\be}{\begin{equation}}
\newcommand{\ee}{\end{equation}}
\newcommand{\ba}{\begin{array}}
\newcommand{\ea}{\end{array}}
\newcommand{\bea}{\begin{eqnarray}}
\newcommand{\eea}{\end{eqnarray}}

\newcommand{\tran}{^{\mbox{\scriptsize T}}}  

\newcommand{\asign}{{\mbox{$\colon\hspace{-2mm}=\hspace{1mm}$}}}
\newcommand{\cf}{{\textit{cf.}}}

\newcommand{\vbar}{\raisebox{.17ex}{\rule{.04em}{1.35ex}}}
\newcommand{\vbarind}{\raisebox{.01ex}{\rule{.04em}{1.1ex}}}
\newcommand{\R}{\ifmmode {\rm I}\hspace{-.2em}{\rm R} \else ${\rm I}\hspace{-.2em}{\rm R}$ \fi}
\newcommand{\D}{\ifmmode {\rm I}\hspace{-.2em}{\rm D} \else ${\rm I}\hspace{-.2em}{\rm D}$ \fi}
\newcommand{\T}{\ifmmode {\rm I}\hspace{-.2em}{\rm T} \else ${\rm I}\hspace{-.2em}{\rm T}$ \fi}
\newcommand{\N}{\ifmmode {\rm I}\hspace{-.2em}{\rm N} \else \mbox{${\rm I}\hspace{-.2em}{\rm N}$} \fi}
\newcommand{\B}{\ifmmode {\rm I}\hspace{-.2em}{\rm B} \else \mbox{${\rm I}\hspace{-.2em}{\rm B}$} \fi}
\newcommand{\Hil}{\ifmmode {\rm I}\hspace{-.2em}{\rm H} \else \mbox{${\rm I}\hspace{-.2em}{\rm H}$} \fi}
\newcommand{\C}{\ifmmode \hspace{.2em}\vbar\hspace{-.31em}{\rm C} \else \mbox{$\hspace{.2em}\vbar\hspace{-.31em}{\rm C}$} \fi}
\newcommand{\Cind}{\ifmmode \hspace{.2em}\vbarind\hspace{-.25em}{\rm C} \else \mbox{$\hspace{.2em}\vbarind\hspace{-.25em}{\rm C}$} \fi}
\newcommand{\Q}{\ifmmode \hspace{.2em}\vbar\hspace{-.31em}{\rm Q} \else \mbox{$\hspace{.2em}\vbar\hspace{-.31em}{\rm Q}$} \fi}
\newcommand{\Z}{\ifmmode {\rm Z}\hspace{-.28em}{\rm Z} \else ${\rm Z}\hspace{-.28em}{\rm Z}$ \fi}

\renewcommand{\vec}[1]{\bf{#1}}     

\newenvironment{Exx}
{\begin{adjustwidth}{0.022\linewidth}{0cm}
\begingroup\small
\vspace{-1.0em}
\raisebox{-.25em}{\rule{\linewidth}{0.3pt}}
\begin{example}
}
{
\end{example}
\vspace{-6mm}
\rule{\linewidth}{0.3pt}
\endgroup
\end{adjustwidth}}

\newenvironment{Algg}
{\begin{adjustwidth}{0.022\linewidth}{0cm}
\begingroup\footnotesize
\vspace{-1.0em}
\raisebox{-.25em}{\rule{\linewidth}{0.3pt}}
\begin{alg}
}
{
\end{alg}
\vspace{-6mm}
\rule{\linewidth}{0.3pt}
\endgroup
\end{adjustwidth}}

\usepackage[dvips]{graphicx}
\usepackage{amsmath}
\usepackage{subfigure}
\usepackage{cite}
\usepackage{url}
\usepackage{rotating}
\usepackage{color}
\usepackage{xcolor}
\usepackage{supertabular}
\usepackage{changepage}

\usepackage{amssymb}
\usepackage{psfrag}
\usepackage{epstopdf}
\usepackage{epsfig}
\usepackage{tikz}
\usetikzlibrary{arrows,decorations.pathmorphing,backgrounds,fit,petri,positioning,matrix}
\usetikzlibrary{shapes,arrows}
\usepackage{multirow}
\usepackage{tkz-euclide}
\usetkzobj{all}
\usepackage{caption}

\newcommand{\add}[1]{{{\color{blue!0!black}#1}}}
\newcommand{\addnew}[1]{{{\color{blue!0!black}#1}}}

\renewcommand{\baselinestretch}{1.275}
\AtBeginDocument{\setlength\abovedisplayskip{2mm}}
\AtBeginDocument{\setlength\belowdisplayskip{1.5mm}}
\allowdisplaybreaks[2]

\usepackage{setspace}

\begin{document} 
\title{\addnew{On the Privacy of Optimization Approaches}}
\author{P.~C. Weeraddana,,~\IEEEmembership{Member,~IEEE,}
\thanks{P. C. Weeraddana, G. Athanasiou, M. Jakobsson, and C. Fischione with Electrical Engineering, KTH Royal Institute of Technology, Stockholm, Sweden. J. S. Baras is with the Department of Electrical and Computer Engineering, University of Maryland, USA. }
G. Athanasiou,~\IEEEmembership{Member,~IEEE,} \\ M. Jakobsson,~\IEEEmembership{Student Member,~IEEE,}  C. Fischione,~\IEEEmembership{Member,~IEEE,} \\ and J. S. Baras,~\IEEEmembership{Fellow,~IEEE}
}


\maketitle

\setlength{\baselineskip}{21pt}

\vspace{-10mm}
\begin{abstract}
\addnew{
Ensuring privacy of sensitive data is essential in many contexts, such as healthcare data, banks, e-commerce, wireless sensor networks, and social networks. It is common that different entities coordinate or want to rely on a third party to solve a specific problem. At the same time, no entity wants to publish its problem data during the solution procedure unless there is a privacy guarantee. Unlike cryptography and differential privacy based approaches, the methods based on optimization lack a quantification of the privacy they can provide. The main contribution of this paper is to provide a mechanism to quantify the privacy of a broad class of optimization approaches. In particular, we formally define a one-to-many relation, which relates a given adversarial observed message to an uncertainty set of the problem data. This relation quantifies the potential ambiguity on problem data due to the employed optimization approaches. The privacy definitions are then formalized based on the uncertainty sets. The properties of the proposed privacy measure is analyzed. The key ideas are illustrated with examples, including localization, average consensus, among others. 



}
\end{abstract}
\vspace{4mm}
\begin{keywords}\vspace{-0mm}
Privacy, distributed optimization, ADMM, secured multiparty computation.
\end{keywords}


\section{Introduction}
Privacy is central in many application domains, where parties jointly solve optimization problems. These interactions and collaborations are desirable to gain mutual benefits. For example, independent hospitals would like to coordinate for diagnostic decision making based on their existing patient records. Normally, optimization solvers require public data sharing, which substantially hinder the cooperation due to privacy concerns (e.g., privacy for patients' records). The challenge is how to \emph{solve} problems among parties, while preserving privacy of their individual data. %
\addnew
{
\subsection{Problem Statement}\label{subsec:problem_statement}
Formally, an $n$-party design and decision making problem can be posed as follows:
\begin{equation} \label{eq:optimization_prob}
\begin{array}{ll}
\mbox{minimize} & f_0({\vec x};{\vec c}_1,\ldots,{\vec c}_N) \\
\mbox{subject to} & { g}({\vec x};{\vec c}_1,\ldots,{\vec c}_N) \preceq {\vec 0} \\
& {h}({\vec x};{\vec c}_1,\ldots,{\vec c}_N)= {\vec 0} \ .
\end{array}
\end{equation}
Here, $f_0(\cdot)$ is a scalar-valued objective function of the decision variables ${\vec x}\in\R^M$, given the problem data $\{{\vec c}_i\}_{i=1:N}$,~\footnote{\addnew{Dimensions of ${\vec c}_i$'s, can also be problem data, which we considered known, throughout to avoid cumbersome notations.}} where ${\vec c}_i$ is the \emph{private} data (e.g., matrices) associated with entity~$i$. Moreover, ${ g}(\cdot)$ and ${ h}(\cdot)$ are vector-valued nonlinear and affine functions, respectively.

We now pose the questions: \texttt{Q1}) How to quantify the privacy of the problem data $\{{\vec c}_i\}_{i=1:N}$~? \texttt{Q2}) How to \emph{solve} problem~\eqref{eq:optimization_prob}, such that the ``privacy" of $\{{\vec c}_i\}_{i=1:N}$ is preserved ? 

}

\subsection{Existing Approaches}

\addnew
{
An ideal answer to \texttt{Q1} is \emph{perfect privacy}, where no information of the original problem data $\{{\vec c}_i\}_{i=1:N}$ can be extracted by a third party. However, answering \texttt{Q2} with perfect privacy is generally impossible. There is an alternative fundamental theory for answering \texttt{Q1}, namely, cryptography~\cite{Goldreich-book-2004}, which includes a number of mechanisms to encrypt the original problem data $\{{\vec c}_i\}_{i=1:N}$. The underlying privacy definition is called \emph{semantic security}~\cite[\S~5]{Goldreich-book-2004}. Loosely speaking, the semantic privacy means that nothing can be \emph{feasibly extracted} from the encrypted data. Given \texttt{Q1} is handled by the semantic security, \emph{secured multiparty computation} (SMC)~\cite{Yao-1982} is a well established framework for answering \texttt{Q2}. In principle, SMC is applied to almost any problem, given that the required computations are written as a Boolean circuit, or as an arithmetic circuit. However, the general objective and constraints of problem~\eqref{eq:optimization_prob} can be complicated functions, and therefore usually the associated solution approach become impractical~\cite{Damgard-06,Toft-2009},\cite[\S~2.3.2]{Bednarz-2012}.
}

\addnew
{
An alternative line of research for answering \texttt{Q1} considers the \emph{$\epsilon$-differential privacy}~\cite{Dwork-McSherry-Nissim-Smith-2006}. The definition allows a non-negligible
information leakage, denoted by $\epsilon$, as opposed to the cryptographic counterpart. Roughly speaking, the standard way of accomplishing $\epsilon$-differential privacy for \texttt{Q2} is based on a random perturbation mechanism, where the underlying subroutines of the associated solution method is appropriately mutated by using random noise with specific statistical properties, \cf~\cite{Chaudhuri-Monteleoni-Sarwate-2011,Kifer-Smith-Thakurta-2012,Guptaetal-2010,Ny-Pappas-CDC-2013,Ny-Pappas-Alerton-2013,Huang-Mitra-Vaidya-2014}. As a consequence, there is a trade-off between the optimality of the solution of problem~\eqref{eq:optimization_prob} and the desired privacy level of the problem data, quantified by $\epsilon$. In addition, the applicability of differential privacy based approaches for solving problem~\eqref{eq:optimization_prob} can usually be limited. Note that specifying the noise properties, such as variances to achieve a pre-specified differential privacy is solely dependent on the \emph{sensitivity} of the related functions (e.g., objective, constraint, subgradient functions) with respect to \emph{perturbations of the problem data}. In general, quantifying such sensitivities in a \emph{closed-form} is challenging. One can certainly compute some bounds. However, unless the bounds are tight, the noise variance computed based on the bounds will be over estimated, which in turn can degrade the optimality of the solution significantly~\cite{Dwork-McSherry-Nissim-Smith-2006,Chaudhuri-Monteleoni-Sarwate-2011,Huang-Mitra-Vaidya-2014}.

}

\addnew{
As opposed to the cryptography and differential privacy based methods discussed above, there is another class of approaches, which answers~\texttt{Q2}~\cite{Mangasarian-OptLet-2011,Mangasarian-OptMethSW-2011,Mangasarian-OptLet-2012,Mangasarian-Wild-2008,Mangasarian-Wild-Fung-2008,Bednarz-2012,Dreier-Kerschbaum-2011,Bednarz-Bean-Roughan-2009,Wang-Ren-Wang-11}, yet answers to~\texttt{Q1} are rather intuitively explained. The key idea of the methods is to use algebraic manipulations to disguise the original problem into an equivalent problem so that the problem data $\{{\vec c}_i\}_{i=1:N}$ is somehow hidden. We refer to these approaches as algebraic transformation methods. Algebraic transformation methods are promising in the sense that they are efficient and usually guarantee the optimality of the solution of problem~\eqref{eq:optimization_prob}, though privacy of the methods is not quantified. The approaches are typically applied to a broader class of problems, as opposed to the cryptographic or differential privacy based approaches. Therefore it is still desirable to quantify their privacy as response to~\texttt{Q1}.
}


\subsection{Our Contributions}

\addnew{The main contribution of this paper is to quantify the privacy properties of a broader class of optimization approaches, where the algebraic transformation methods  are particular cases. More specifically, the proposed privacy definition applies to 1) methods via standard objective/constraint transformations, 2) methods via standard variable transformations, 3) classical decomposition methods, e.g., primal and dual decomposition, and 4) state-of-the-art alternating direction method of multipliers method (ADMM).}

\addnew{To do this, the underlying disguise of problem data $\{{\vec c}_i\}_{i=1:N}$ due to those optimization approaches is formally represented as a one-to-many binary relation between the observed message of an adversary and an uncertainty set, in which the sensitive problem data $\{{\vec c}_i\}_{i=1:N}$ itself is a member. The privacy of the associated method can then be quantified by the properties of the uncertainty set. The properties of the proposed privacy measure is analyzed. Given an optimization approach as response to \texttt{Q2}, our definition answers \texttt{Q1}. Implications of the proposed privacy measure is discussed. Our privacy definition allows comparing the privacy of different optimization approaches applied to a problem. It quantifies the \emph{best} and the \emph{worst} privacy for an specified problem data. Several examples are given to illustrate the key ideas.}

In \S~\ref{sec:definitions} we present our new proposed definitions and other basic ones, useful for describing the privacy properties of optimization approaches. Application of the proposed privacy definitions for optimization approaches are illustrated in \S~\ref{sec:Opt_Approaches_Privacy}. Conclusions are given in~\S~\ref{sec:conclusions}.

\addnew{\emph{{Notations:}} Boldface lower case and upper case letters represent vectors and matrices, respectively, and calligraphy letters represent sets. The Euclidean-$n$ space and the positive integers are denoted by $\R^n$ and $\N$, respectively. The superscript $(\cdot)^{\tran}$ is the transpose. We use the notation $({\vec a}, {\vec b})$ to denote $[{\vec a}\tran \ {\vec b}\tran]\tran$. The identity matrix, all zero vector, and all one vector, are denoted by ${\vec I}$, ${\vec 0}$, and ${\vec 1}$, respectively. The $\ell_2$-norm of ${\vec x}$ is denote by $||{\vec x}||_2$. The power set of an arbitrary space $\Omega$ is denoted by $\mathcal{P}(\Omega)$. A ball with radius $r$ and center ${\vec b}$ is denoted by $\mathcal{B}(r,{\vec b})$. Finally, we denote by $\succeq$, the component-wise inequality. }

\section{Privacy Quantifications}\label{sec:definitions}

\addnew{Original definitions together with other basic definitions and assumptions for quantifying the privacy properties of the optimization approaches are given. The essential implications of the definitions are also discussed.}

\addnew{
\subsection{Problem Data and the Adversary}\label{subsec:definitions_data_advsry}
}

\addnew{
\begin{defin}[Inputs, input tuples, and messages]\label{def:input}
Consider the problem~(\ref{eq:optimization_prob}). We call the \emph{set of problem data} $\{{\vec c}_i\}_{i=1:N}$, denoted $\mathcal{C}$, the \emph{inputs} of problem~(\ref{eq:optimization_prob}). Suppose ${\vec c} = {\vec f}(({\vec c}_i)_{i=1:N})$, where ${\vec f}$ is a possibly vector valued function with non-constant components. Moreover, suppose that ${\vec c}$ lies in a metric space $\mathcal{X}_{\vec c}$. We call ${\vec c}$, an \emph{input tuple} of $\mathcal{C}$. Finally, we call a finite length data structure ${S}\big(({\vec c}_i)_{i=1:N}\big)=\Big(S_1\big(({\vec c}_i)_{i=1:N}\big),\ldots,S_M\big(({\vec c}_i)_{i=1:N}\big)\Big)$, a \emph{message} of $\mathcal{C}$, where $S_j\big(({\vec c}_i)_{i=1:N}\big)$ represents any vector/matrix valued function. \hfill $\blacktriangleleft$
\end{defin}
}
\addnew{Note that elements of ${\vec c}$ represents the central components what the problem data owner cares about. The key points of the definition is illustrated in Example~\ref{ex:input_output1}. Throughout the paper we consider the following assumption on the input of a problem:}
\addnew{
\begin{assump}\label{assump:input}
The inputs of any optimization problem of the form~(\ref{eq:optimization_prob}) is deterministic, and therefore there is no associated statistical models. \hfill $\blacktriangleleft$
\end{assump}
}
The considered adversarial model and its associated knowledge is defined next. The definition of the adversary is similar to the passive eavesdroppers considered in~\cite[\S~5.1-5.3]{Goldreich-book-2004}. We consider the \emph{passive adversary} model throughout this paper.
%

\begin{defin}[\add{Passive adversary}]\label{def:adversary}
\add{In a multi-party environment, a party involved in solving a problem of the form~(\ref{eq:optimization_prob}), or even a third party, is called a passive adversary if it taps the messages of input of (\ref{eq:optimization_prob}) exchanged during different stages of the solution method, keeps a record of possibly all the messages exchanged, and tries to discover others' private~data.}  \hfill $\blacktriangleleft$
\end{defin}


\addnew{
\begin{defin}[Adversarial knowledge]\label{def:AdvsryKnowledge}
The set $\mathcal{K}$ of information that an adversary might exploit to discover the input of problem~(\ref{eq:optimization_prob}) is called the adversarial knowledge. The knowledge can contain messages $S$ of the problem input, which we call the \emph{inevitable} knowledge. The rest of the knowledge $\mathcal{K}\setminus\{S\}$ is called the \emph{auxiliary} knowledge. \hfill $\blacktriangleleft$
\end{defin}
}

\addnew{
Note that the inevitable knowledge contains messages (e.g., $S$, \cf~Definition~\ref{def:input}) that must be revealed to a third party, such as transformed variants of inputs necessary for executing the solution method. The auxiliary knowledge can encompass components, such as eavesdropped measures of input elements, knowledge of optimality properties, etc.,~\cf~Example~\ref{ex:input_output1}. The following remark summarizes a property of $\mathcal{K}$ that will be useful later.

\begin{remark}\label{remark:knowledge}
For a given problem and its solution approach, the knowledge $\mathcal{K}$ of a third party can be further improved and can never be decreased as time elapses. In other words, given time $t_0\leq t_1\leq t_2\leq \cdots$, we have $\mathcal{K}=\mathcal{K}_{t_0}\subseteq\mathcal{K}_{t_1}\subseteq\mathcal{K}_{t_2}\subseteq\cdots$, where $\mathcal{K}_{t_i}$ denotes the adversarial knowledge at time $t_i$. As a result, unless the underlying solution method is changed, it is always the case that $\mathcal{K}_{t_0}\subseteq\mathcal{K}_{t_1}\subseteq\mathcal{K}_{t_2}\subseteq\cdots$. \hfill $\blacktriangleleft$
\end{remark}

Note that \emph{different} solution methods in general can yield different sets of knowledge, even if the problem is fixed. Finally, recall that the set $\mathcal{C}$ of the inputs/outputs is deterministic, \cf~Assumption~\ref{assump:input}. Therefore, the adversarial knowledge $\mathcal{K}$ does not contain elements such as probability density functions or cumulative distribution functions to describe statistical models for ${\vec c}_i\in\mathcal{C}$.}


\addnew
{
\begin{figure}[t]
\centering
\begin{tikzpicture}[scale=0.30]
\node (C) at (-3,0) {\scriptsize{$\{S\}$}};
\node[ellipse,draw=black,fill=blue!20,minimum height=10mm,minimum width=60mm] (A) at (12,0) {\tiny{$\mathcal{U}_{\vec c}(\mathcal{K}_1)$}};
\node[ellipse,draw=black,fill=white,minimum height=3mm,minimum width=3mm] (B) at (16.7,-0.2) {\tiny{$\mathcal{U}_{\vec c}(\mathcal{K}_2)$}};
\draw (C) edge [-latex, thick, out=30,in=160]  node[above] {\tiny{$R_{\vec c}(\mathcal{K}_1)$}} (A) ;
\draw (C) edge [-latex, thick, out=-30,in=190]  node[below] {\tiny{$R_{\vec c}(\mathcal{K}_2)$}} (B) ;
\end{tikzpicture}
\captionof{figure}{\label{Fig:relation}\small{\addnew{Obfuscation of an input tuple ${\vec c}$, where $S\subseteq\mathcal{K}_1$, $S\subseteq\mathcal{K}_2$, $\mathcal{K}_1\subseteq\mathcal{K}_2$, and the ellipsoids represent uncertainty sets of ${\vec c}$.}}}
\vspace{-7mm}
\end{figure}
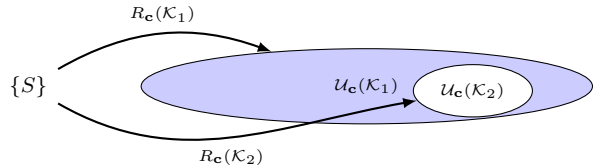
}

\addnew{
\begin{defin}[Obfuscation of the input and uncertainty sets]\label{def:Obsf_Unctnty}
Let $\mathcal{C}$, ${\vec c}$, ${S}$, and $\mathcal{K}$ denote the input, an input tuple, a message of problem~(\ref{eq:optimization_prob}), and the adversarial knowledge, respectively, \cf~Definition~\ref{def:input}. Suppose ${S}\subseteq\mathcal{K}$. We say ${\vec c}\in \mathcal{X}_{\vec c}$ is obfuscated if there exists a \emph{one-to-many} \emph{binary relation} $\mathcal{R}_{\vec c}$ from the singleton $\{S\}$ to a nonempty set $\mathcal{U}_{\vec c}\subseteq \mathcal{X}_{\vec c}$, such that $(S,{\vec c})\in \mathcal{R}_{\vec c}$. We call the (possibly unbounded) set $\mathcal{U}_{\vec c}=\{{\vec u} \ | \  (S,{\vec u})\in \mathcal{R}_{\vec c} \}$, the uncertainty set for ${\vec c}$. \hfill $\blacktriangleleft$
\end{defin}
}

\addnew
{
The \emph{one-to-many} nature of such relations is the mechanism of achieving ambiguity in ${\vec c}$, \cf~Figure~\ref{Fig:relation}. An example to clarify Definition~\ref{def:Obsf_Unctnty} is given below:
\begin{Exx}\label{ex:input_output1}
Consider a client-cloud environment, where the client has an optimization problem with problem input $\mathcal{C}=\{{\vec a},{\vec b}\}$, where ${\vec a},{\vec b}\in\R^n$ are private data that the client cares about. The cloud (potential adversary) solves the problem for the client. Suppose that the employed deterministic solution methods for the problem is such that client has to outsource the two scalars $s_1$ and $s_2$ to the cloud, where $s_1=||{\vec a}||^2_2$ and $s_2={\vec b}\tran{\vec a}$, that is the message $S=(s_1,s_2)$. Suppose, in addition, cloud knows $n$ (i.e., auxiliary knowledge) and wants to find input tuple ${\vec a}$. Then we have $\mathcal{R}_{\vec a}=\{(S,{\vec a}) \ | \   s_1=||{\vec a}||^2_2, s_2={\vec b}\tran{\vec a}, {\vec b}\in\R^n , {\vec a}\in\R^n\}$.
\end{Exx}
}

\addnew{
Trivially, having more knowledge of an input tuple ${\vec c}$ can only reduce the ambiguity in ${\vec c}$. This phenomenon is formally expressed in the following remark, \cf~Figure~\ref{Fig:relation}.
\begin{remark}\label{remark:relation}
Let $\mathcal{R}_{\vec c}(\mathcal{K}_1)$ and $\mathcal{R}_{\vec c}(\mathcal{K}_2)$ be relations defined as in Definition~\ref{def:Obsf_Unctnty} associated with a problem, where the \emph{inevitable} knowledge is identical in both cases, i.e., $\mathcal{K}_1\supseteq S$ and $\mathcal{K}_2\supseteq S$ for some message $S$. Then $\mathcal{K}_1\subseteq\mathcal{K}_2\Rightarrow\mathcal{U}_{\vec c}(\mathcal{K}_1)\supseteq \mathcal{U}_{\vec c}(\mathcal{K}_2)$. \hfill $\blacktriangleleft$
\end{remark}
}

\addnew{We note that, in many interesting cases, the considered optimization approaches in this paper can be used to realize a relation $\mathcal{R}_{\vec c}$ as defined in Definition~\ref{def:Obsf_Unctnty}. We capitalize on the properties of such relations resulted by optimization approaches, when quantifying the privacy.}

\addnew{
\subsection{Privacy Index}\label{subsec:privacyIndex}
}

\addnew{Given a relation $\mathcal{R}$ and its associated uncertainty set $\mathcal{U}$ that lies in a metric space $\mathcal{X}$ (\cf~Definition~\ref{def:Obsf_Unctnty}). We consider the following metrics to quantify the size and the spread of $\mathcal{U}$.
\begin{enumerate}
\item Diameter of $\mathcal{U}$, denoted \texttt{d}($\mathcal{U}$).
\item Counting measure of $\mathcal{U}$, denoted by $\mu(\mathcal{U})$~\cite[p.~146]{McDonald-Weiss-2013}.
\item Affine dimension of $\mathcal{U}$, denoted \texttt{a}($\mathcal{U}$)~\cite[\S~2.1.3]{Boyd-Vandenberghe-04}.
\end{enumerate}
}

\addnew{
The diameter gives a measure of the worst discrepancy between any two points in the uncertainty set. In particular, we define \texttt{d}$(\cdot)$ as follows:
\begin{defin}[Diameter/center of an uncertainty set]\label{def:adversary_attackmodel}
Given an input tuple ${\vec c}$ of the input $\mathcal{C}$ and ${\vec c}$ lies in a metric space $\mathcal{X}_{\vec c}$. Let $\mathcal{U}_{\vec c}\subseteq\mathcal{X}_{\vec c}$ denote the uncertainty set of ${\vec c}$. Then the diameter \texttt{d}~$: \mathcal{X}_{\vec c}\rightarrow\R \cup \infty$ is given~by
\be\label{eq:diameter}
\hspace{-0mm}\texttt{d}(\mathcal{U}_{\vec c}){=} \left\{\hspace{-1mm} \begin{array}{ll}
   2 \ \inf_{\vec v} \sup_{{\vec u}\in\mathcal{U}_{\vec c}} ||\vec v-u|| & \  \textrm{$\mathcal{U}_{\vec c}$ is bounded}\\
  \infty  & \ \mbox{otherwise} \ ,
   \end{array} \right.
\ee
where $||\cdot||$ is some norm on $\mathcal{X}_{\vec c}$. The argument ${\vec v}$ that achieves $\texttt{d}(\mathcal{U}_{\vec c})$ is called the \emph{center} of $\mathcal{U}_{\vec c}$ and is denoted by $\texttt{c}(\mathcal{U}_{\vec c})$. If the set $\mathcal{U}_{\vec c}$ is unbounded, we say any point in $\mathcal{X}_{\vec c}$ achieves $\texttt{d}(\mathcal{U}_{\vec c})$, i.e., $\texttt{c}(\mathcal{U}_{\vec c})$ is arbitrary.
\hfill $\blacktriangleleft$
\end{defin}

It is not difficult to see that the ball $\mathcal{B}(\texttt{d}(\mathcal{U}_{\vec c})/2,\texttt{c}(\mathcal{U}_{\vec c}))$, denoted $\mathcal{B}_{\mathcal{U}_{\vec c}}$, is the smallest volume ball, which encloses $\mathcal{U}_{\vec c}$. The metric $\texttt{d}(\mathcal{U}_{\vec c})$ for measuring the size of the uncertainty set $\mathcal{U}_{\vec c}$ can be interpreted and motivated as follows. Suppose the set $\mathcal{U}_{\vec c}$ be bounded and an adversary wants to estimate ${\vec c}$. Because ${\vec c}$ is deterministic (\cf~Assumption~\ref{assump:input}), a reasonable criterion, from the view of the adversary, is to pick ${\vec c}^\star_{\textrm{adv}}$ that minimizes the associated worst-case
error $||{\vec c}_{\textrm{adv}}-{\vec c}||$ of a candidate guess ${\vec c}_{\textrm{adv}}$. These strategies are known as \emph{worst-case robust designs}, \cf \cite[\S~6.4.2]{Boyd-Vandenberghe-04}. From \eqref{eq:diameter}, it follows that ${\vec c}^\star_{\textrm{adv}}=\texttt{c}(\mathcal{U}_{\vec c})$ and the adversary's worst discrepancy between its guess ${\vec c}^\star_{\textrm{adv}}$ and ${\vec c}$ is given by $(1/2)\texttt{d}(\mathcal{U}_{\vec c})$. On the other hand, if the set $\mathcal{U}_{\vec c}$ is unbounded, a rational adversary has no preference for one point in $\mathcal{X}_{\vec c}$ over another. Therefore, the adversary's worst discrepancy is unbounded above, irrespective of its guess ${\vec c}^\star_{\textrm{adv}}$.

The counting measure of an uncertainty set is an estimate of the number elements in it~\cite[p.~146]{McDonald-Weiss-2013}, and is defined below.
\begin{defin}[Counting measure of an uncertainty set]\label{def:Cardinality}
Given an input tuple ${\vec c}$ of the input $\mathcal{C}$ that lies in a metric space $\mathcal{X}_{\vec c}$. Let $\mathcal{U}_{\vec c}\subseteq\mathcal{X}_{\vec c}$ denote the uncertainty set of ${\vec c}$. The counting measure of $\mathcal{U}_{\vec c}$ is $\mu(\mathcal{U}_{\vec c})$, where $\mu{:}\mathcal{P}(\mathcal{X}_{\vec c}){\rightarrow} \N {\cup} \infty$ is
\be\label{eq:cardinality}
\mu(\mathcal{A})= \left\{ \begin{array}{ll}
  |\mathcal{A}| & \ \ \textrm{$\mathcal{A}$ is finite}\\
  \infty  & \ \ \mbox{otherwise} \ ,
   \end{array} \right.
\ee
with $|\cdot|$ denoting the cardinality of a \emph{finite} set. \hfill $\blacktriangleleft$
\end{defin}

The metric $\mu(\mathcal{U}_{\vec c})$ for quantifying the uncertainty set $\mathcal{U}_{\vec c}$ can be motivated as follows. If $\mathcal{U}_{\vec c}$ is finite, the value $\mu(\mathcal{U}_{\vec c})$ can be used to quantify the adversary's \emph{best} worst-case winning~\footnote{\addnew{Here the term \emph{winning} refers to the event of  guessing correctly ${\vec c}$.}} probability, denoted $p^{\textrm{win}}_{\textrm{w-c}}$, and to characterizes the probability mass function (pmf) of adversary's guess, denoted ${\vec p}^\star_{\textrm{adv}}$, to achieve $p^{\textrm{win}}_{\textrm{w-c}}$. Recall that ${\vec c}$ is deterministic,~\cf~Assumption~\ref{assump:input}. That is, the adversary has no statistical models for describing the variation in ${\vec c}$. Therefore, a natural criterion for any rational adversary is to choose a probability distribution that maximizes the resulting winning probability $\texttt{Pr}(\textrm{adversary wins})$, among all possible probability distributions for ${\vec c}$. From basic probability rules, we get $p^{\textrm{win}}_{\textrm{w-c}}= 1/\mu(\mathcal{U}_{\vec c})$ and ${\vec p}^\star_{\textrm{adv}}=(1/\mu(\mathcal{U}_{\vec c})){\vec 1}$, as probabilistic intuition requires. If $\mathcal{U}_{\vec c}$ is infinite, as our probabilistic intuition suggests, $p^{\textrm{win}}_{\textrm{w-c}}$ ought to be $0$. Thus, $\mu(\mathcal{U}_{\vec c})$ quantifies the ``practical impossibility" of an \emph{exact} guess.

%
%

The last metric used to quantify the properties of $\mathcal{U}_{\vec c}$ is its affine dimension, i.e., $\texttt{a}(\mathcal{U}_{\vec c})$. In particular, $\texttt{a}(\mathcal{U}_{\vec c})$ is the dimension of its affine hull~\cite[\S~2.1.3]{Boyd-Vandenberghe-04}. Roughly speaking, $\texttt{a}(\mathcal{U}_{\vec c})$ quantifies the number of orthogonal directions, along which an adversary has to perform its search for locating ${\vec c}$.
}

\addnew{
Let us finally give a formal definition to quantify the privacy of optimization methods considered in this paper. To do this, we use the metrics $\texttt{d}(\cdot)$, $\mu(\cdot)$, and $\texttt{a}(\cdot)$ discussed above.}
\addnew{
\begin{defin}[Privacy index]\label{def:input_privacy}
Let $\mathcal{C}$ denote the input of problem~(\ref{eq:optimization_prob}) and ${\vec c}$ denote any input tuple of $\mathcal{C}$ that lies in a metric space $\mathcal{X}_{\vec c}$. Moreover, let $\mathcal{U}_{\vec c}\subseteq\mathcal{X}_{\vec c}$ denote the uncertainty set of ${\vec c}$. Then we call the 3-component tuple $(\texttt{d}(\mathcal{U}_{\vec c}),\nu(\mathcal{U}_{\vec c}),\texttt{a}(\mathcal{U}_{\vec c}))$ the privacy index of ${\vec c}$, where $\nu(\mathcal{U}_{\vec c})=1-1/\mu(\mathcal{U}_{\vec c})$. Moreover, the corresponding tuple is denoted by $\boldsymbol{\rho}(\mathcal{U}_{\vec c})$, i.e., $\boldsymbol{\rho}(\mathcal{U}_{\vec c})=(\texttt{d}(\mathcal{U}_{\vec c}),\nu(\mathcal{U}_{\vec c}),\texttt{a}(\mathcal{U}_{\vec c}))$. \hfill $\blacktriangleleft$
\end{defin}
}
\addnew
{
\noindent Note that $\nu(\mathcal{U}_{\vec c})\in[0,1)$ is the adversary's \emph{best} worst-case loosing probability, i.e., $\nu(\mathcal{U}_{\vec c})=1-p_{\textrm{w-c}}^{\textrm{win}}$. Some basic properties of the $\boldsymbol{\rho}(\cdot)$ are proved in the following proposition.
\begin{prop}\label{prop:privacy_index}
The privacy index $\boldsymbol{\rho}(\mathcal{U})$ has the following properties:
\begin{enumerate}
\item[a)] ${\vec 0}\preceq\boldsymbol{\rho}(\mathcal{U})\preceq(\infty,1,n)$ for all $\mathcal{U}\subseteq\R^n$.
\item[b)] $\mathcal{U}$ is a singleton $\Leftrightarrow$ $\boldsymbol{\rho}(\mathcal{U})={\vec 0}$.
\item[c)] $\boldsymbol{\rho}({\vec v} + \mathcal{U})= \boldsymbol{\rho}(\mathcal{U})$, for all ${\vec v}\in\R^n$ and $\mathcal{U}\subseteq\R^n$, where ${\vec v}+\mathcal{U}=\{{\vec v}+{\vec u} \ | \ {\vec u}\in\mathcal{U}\}$.
\item[d)] $\boldsymbol{\rho}({\vec P}\mathcal{U})= \boldsymbol{\rho}(\mathcal{U})$, for all unitary matrices ${\vec P}\in\R^{n\times n}$ and $\mathcal{U}\subseteq\R^n$, where ${\vec P}\mathcal{U}=\{{\vec P}{\vec u} \ | \ {\vec u}\in\mathcal{U}\}$.
\item[e)] $\mathcal{U}\subseteq\mathcal{V}\Rightarrow \boldsymbol{\rho}(\mathcal{U})\preceq\boldsymbol{\rho}(\mathcal{V})$.
\item[f)] Let $\mathcal{U}_{\vec c},\mathcal{V}_{\vec c}\subseteq\R^n$ denote two uncertainty sets of ${\vec c}$. Then $\boldsymbol{\rho}(\mathcal{U}_{\vec c})\preceq\boldsymbol{\rho}(\mathcal{V}_{\vec c})$ $\Rightarrow$ $\exists~{\mathcal{U}}^\star,{\mathcal{V}}^\star\in\R^n$, containing ${\vec c}$ such that $\boldsymbol{\rho}(\mathcal{U}_{\vec c})=\boldsymbol{\rho}({\mathcal{U}}^\star)\preceq\boldsymbol{\rho}({\mathcal{V}}^\star)=\boldsymbol{\rho}(\mathcal{V}_{\vec c})$, $\mathcal{B}_{\mathcal{U}^\star}\subseteq\mathcal{B}_{\mathcal{V}^\star}$. \hfill $\blacktriangleleft$
\end{enumerate}

\end{prop}
\begin{proof}

\textit{a,b)} The proof of these parts follow directly from the definition of the metrics $\texttt{d}$, $\mu$, and $\texttt{a}$.

\textit{c,d)} Theses parts follows from that translation and rotation do not change sets in $\R^n$.

\textit{e)} Let us first show that $\texttt{d}(\mathcal{U})\leq \texttt{d}(\mathcal{V})$. If at least one set of $\mathcal{U}$ and $\mathcal{V}$ is unbounded, then $\texttt{d}(\mathcal{U})\leq \texttt{d}(\mathcal{V})$ holds trivially because for any unbounded $\mathcal{W}$, $\texttt{d}(\mathcal{W})=\infty$. If $\mathcal{U}$ and $\mathcal{V}$ are bounded, for any ${\vec x}\in\R^n$, $g_{\mathcal{U}}({\vec x})=\sup_{{\vec u}\in\mathcal{U}}||{\vec x}-{\vec u}||$ is bounded, so is $g_{\mathcal{V}}({\vec x})=\sup_{{\vec v}\in\mathcal{V}}||{\vec x}-{\vec v}||$. Moreover, if $\mathcal{U}\subseteq\mathcal{V}$, then $g_{\mathcal{U}}({\vec x})\leq g_{\mathcal{V}}({\vec x})$ for all ${\vec x}\in\R^n$. Therefore, $\inf_{{\vec x}}g_{\mathcal{U}}({\vec x})=\texttt{d}(\mathcal{U})\leq \inf_{{\vec x}}g_{\mathcal{V}}({\vec x})=\texttt{d}(\mathcal{V})$ as required. The relation $\nu(\mathcal{U})\leq \nu(\mathcal{V})$ follows from that $\mu(\mathcal{U})\leq \mu(\mathcal{V})$. In particular, $(\R^n,\mathcal{P}(\R^n),\mu)$, $\mu$ being the counting measure on $\mathcal{P}(\R^n)$, is a measure space, \cf~\cite[p.~146]{McDonald-Weiss-2013}. Therefore, from \cite[Theorem~5.1]{McDonald-Weiss-2013}, $\mathcal{U}\subseteq\mathcal{V}\Rightarrow \mu(\mathcal{U})\leq\mu(\mathcal{V})$, i.e., $\nu(\mathcal{U})\leq\nu(\mathcal{V})$. Finally, $\texttt{a}(\mathcal{U})\leq\texttt{a}(\mathcal{V})$ follows by definition.

\textit{f)} Let $\bar{\mathcal{U}}=\{{\vec P}_{\mathcal{U}} ({\vec u}+{\vec t}_{\mathcal{U}}) \ | \ {\vec u}\in\mathcal{U}_{\vec c}\}$ and $\bar{\mathcal{V}}=\{{\vec P}_{\mathcal{V}}({\vec v}+{\vec t}_{\mathcal{{\vec v}}}) \ | \ {\vec v}\in\mathcal{V}_{\vec c}\}$ for some ${\vec t}_{\mathcal{U}},{\vec t}_{\mathcal{V}}\in\R^n$ and unitary matrices ${\vec P}_{\mathcal{U}},{\vec P}_{\mathcal{V}}\in\R^{n\times n}$. Then from the hypothesis and from part c), d), it follows that $\boldsymbol{\rho}(\mathcal{U}_{\vec c})=\boldsymbol{\rho}(\bar{\mathcal{U}})\preceq\boldsymbol{\rho}(\bar{\mathcal{V}})=\boldsymbol{\rho}(\mathcal{V}_{\vec c})$. Let us now find a particular choice ${\vec t}^\star_{\mathcal{U}},{\vec t}^\star_{\mathcal{V}}, {\vec P}^\star_{\mathcal{U}}$, and ${\vec P}^\star_{\mathcal{V}}$ to compute $\mathcal{U}^\star$ and $\mathcal{V}^\star$. If both $\mathcal{U}_{\vec c},\mathcal{V}_{\vec c}$ are unbounded, then choose ${\vec t}^\star_{\mathcal{U}}=0,{\vec t}^\star_{\mathcal{V}}=0, {\vec P}^\star_{\mathcal{U}}={\vec P}^\star_{\mathcal{V}}={\vec I}$. The result follows immediately because $\mathcal{B}_{\mathcal{U}^\star}= \mathcal{B}_{\mathcal{V}^\star}=\R^n \cup \boldsymbol{\infty}$. If only $\mathcal{V}_{\vec c}$ is unbounded, the same choice of ${\vec t}^\star_{\mathcal{U}},{\vec t}^\star_{\mathcal{V}}, {\vec P}^\star_{\mathcal{U}}$, and ${\vec P}^\star_{\mathcal{V}}$ yields the result, because $\mathcal{B}_{\mathcal{V}^\star}=\R^n \cup \boldsymbol{\infty}$. However, if both $\mathcal{U}_{\vec c},\mathcal{V}_{\vec c}$ are bounded, we have to be more specific. Recall that $\mathcal{B}_{\mathcal{U}_{\vec c}}$ is the smallest volume (bounded) ball containing $\mathcal{U}_{\vec c}$, \cf~Definition~\ref{def:adversary_attackmodel}. Therefore, $\exists$ a limit point ${\vec l}_{\mathcal{U}}$ at the boundary of $\mathcal{B}_{\mathcal{U}_{\vec c}}$ such that $\forall~\boldsymbol{\epsilon}\succ{\vec 0}$, $\mathcal{B}(\boldsymbol{\epsilon},{\vec l}_{\mathcal{U}})\cap\mathcal{U}_{\vec c}$ is nonempty, \cf~Figure~\ref{Fig:Prop-f}. Similarly, $\exists$ a limit point ${\vec l}_{\mathcal{V}}$ at the boundary of $\mathcal{B}_{\mathcal{V}_{\vec c}}$ such that $\forall~\boldsymbol{\epsilon}\succ{\vec 0}$, $\mathcal{B}(\boldsymbol{\epsilon},{\vec l}_{\mathcal{V}})\cap\mathcal{V}_{\vec c}$ is nonempty. Then by construction, $\mathcal{U}^\star =\{{\vec u} + ({\vec c}-{\vec l}_{\mathcal{U}}) \ | \ {\vec u}\in\mathcal{U}_{\vec c} \}$ and $\mathcal{V}^\star =\{{\vec P}{\vec v} + ({\vec c}-{\vec P}{\vec l}_{\mathcal{V}}) \ | \ {\vec v}\in\mathcal{V}_{\vec c} \}$ satisfy the required necessary conditions, where ${\vec P}$ is a unitary matrix such that $({\vec a}_{\mathcal{U}}-{\vec l}_{\mathcal{U}})/||{\vec a}_{\mathcal{U}}-{\vec l}_{\mathcal{U}}||_2={\vec P}({\vec a}_{\mathcal{V}}-{\vec l}_{\mathcal{V}})/||{\vec a}_{\mathcal{V}}-{\vec l}_{\mathcal{V}}||_2$, ${\vec a}_{\mathcal{U}}=\texttt{c}(\mathcal{U}_c)$, and ${\vec a}_{\mathcal{V}}=\texttt{c}(\mathcal{V}_c)$, \cf~Figure~\ref{Fig:Prop-f}. To see this, note that the sets $\mathcal{U}^\star,\mathcal{V}^\star$ are based on rotations and translations of the original sets $\mathcal{U}_{\vec c},\mathcal{V}_{\vec c}$. In particular, ${\vec t}^\star_{\mathcal{U}} = {\vec c}-{\vec l}_{\mathcal{U}}$, ${\vec t}^\star_{\mathcal{V}}={\vec c}-{\vec P}{\vec l}_{\mathcal{V}}$, ${\vec P}^\star_{\mathcal{U}} = {\vec I}$, ${\vec P}^\star_{\mathcal{V}} = {\vec P}$. When ${\vec u}={\vec l}_{\mathcal{U}}$, we conclude that ${\vec c}\in\mathcal{U}^\star$ and, so is ${\vec c}\in\mathcal{V}^\star$. Moreover, the rotation induced by ${\vec P}$ ensures $\mathcal{B}_{\mathcal{U}^\star}\subseteq\mathcal{B}_{\mathcal{V}^\star}$, \cf~Figure~\ref{Fig:Prop-f}.

\noindent This completes the proof of the Proposition.
\end{proof}
}

\addnew
{
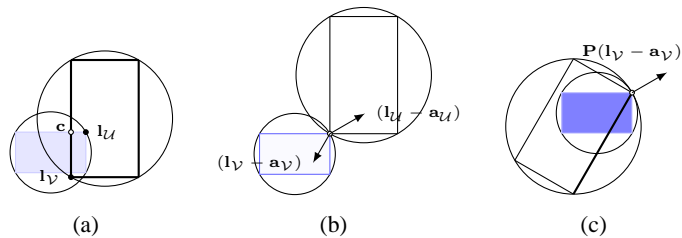
\begin{figure}[t]
\centering
\subfigure[]{
\begin{tikzpicture}[scale=.9]
\draw (-0.8,-0.5) circle (0.6cm);
\fill [color=blue!10,draw=blue!20] (-1.3196, -0.8000) rectangle ( -0.2804, -0.2000);
\draw [thick] (-0.5, 0.8660) -- (-0.5, -0.8660);
\draw [thick] (-0.5, -0.8660) -- (0.5, -0.8660);
\draw [thick] (0.5, -0.8660) -- (0.5, 0.8660);
\draw [thick] (0.5, 0.8660) -- (-0.5, 0.8660);

\draw (0,0) circle (1cm);
\node (a) at (-0.65,-0.13) {\tiny{${\vec c}$}};
\fill [color=white,draw=black] (-0.5,-0.2) circle (1.0pt);
\fill [draw=black] (-0.5, -0.8660) circle (1.0pt) node [left]{\tiny{${\vec l}_{\mathcal{V}}$}};
\fill [draw=black] (-0.2804, -0.2000) circle (1.0pt) node [right]{\tiny{${\vec l}_{\mathcal{U}}$}};
\end{tikzpicture}
\label{fig:MinmaxVsSum-b}}
\subfigure[]{
\begin{tikzpicture}[scale=.9]
\draw (-1.0196,-0.5) circle (0.6cm);
\fill [color=blue!2,draw=blue!60] ( -1.5392, -0.8000) rectangle ( -0.5, -0.2000);
\draw [thin] (-0.5, 1.532) -- (-0.5, -0.20);
\draw [thin] (-0.5, -0.20) -- (0.5, -0.20);
\draw [thin] (0.5, -0.20) -- (0.5, 1.532);
\draw [thin] (0.5, 1.532) -- (-0.5, 1.532);

\draw (0,0.666) circle (1cm);
\fill [color=white,draw=black] (-0.5,-0.2) circle (1.0pt);
\draw [-latex] (-0.5, -0.2) -- (0.0196, 0.1000) node [right] {\tiny{$({\vec l}_{\mathcal{U}}-{\vec a}_{\mathcal{U}})$}};
\draw [-latex] (-0.5, -0.2) -- (-0.75, -0.633) node [left] {\tiny{$({\vec l}_{\mathcal{V}}-{\vec a}_{\mathcal{V}})$}};
\end{tikzpicture}
\label{fig:MinmaxVsSum-a}}
\subfigure[]{
\begin{tikzpicture}[scale=.9]
\draw (-1.0196,-0.5) circle (0.6cm);
\fill [color=blue!50,draw=blue!20] ( -1.5392, -0.8000) rectangle ( -0.5, -0.2000);
\draw [thick] ( -1.3660,  -1.7) -- (-0.5, -0.20);
\draw [thin] (-0.5, -0.20) -- (-1.3660, 0.3);
\draw [thin] (-1.3660, 0.3) -- (-2.2320, -1.2);
\draw [thin] (-2.2320, -1.2) -- (-1.3660,  -1.7);

\draw (-1.3660, -0.7) circle (1cm);
\fill [color=white,draw=black] (-0.5,-0.2) circle (1.0pt);
\node (a) at (-.5, 0.4) {\tiny{${\vec P}({\vec l}_{\mathcal{V}}-{\vec a}_{\mathcal{V}})$}};
\draw [-latex] (-0.5, -0.2) -- (0.0196, 0.1000) ;
\end{tikzpicture}
\label{fig:MinmaxVsSum-a}}
\captionof{figure}{\label{Fig:Prop-f}\small{\addnew{Proposition~\ref{prop:privacy_index}-(f): (a)~Original sets $\mathcal{U}_{\vec c}$ (shaded box) and $\mathcal{V}_{\vec c}$ (solid rectangular wire frame); (b) A representation of intermediate steps, where the $\mathcal{U}_{\vec c}$ and $\mathcal{V}_{\vec c}$ are shifted, yet they still contains ${\vec c}$; (c)~Sets $\mathcal{U}^\star$ (shaded box) and $\mathcal{V}^\star$ (solid rectangular wire frame).}}}
\vspace{-8mm}
\end{figure}
}

%

%

\addnew{As we will see next, Proposition~\ref{prop:privacy_index} is useful for explaining interesting implications of the privacy index. Especially, the parts a), b) e), and f) of Proposition~\ref{prop:privacy_index} are more relevant.}

\addnew
{
\subsection{Implications of the Privacy Index}\label{subsec:privacyIndex}
}

\addnew
{
By definition, the privacy index $\boldsymbol\rho$ lies in the extended $\R^3$ and is nonnegative, \cf~Proposition~\ref{prop:privacy_index}-a). Moreover, Proposition~\ref{prop:privacy_index}-a) provides the \emph{smallest} and the \emph{largest} possible values of $\boldsymbol\rho$ for an input-tuple in $\R^n$. Note that if an uncertainty set $\mathcal{U}_{\vec c}$ is a singleton, then it contains ${\vec c}$ only, \cf~Definition~\ref{def:Obsf_Unctnty}. Therefore, from Proposition~\ref{prop:privacy_index}-b), $\boldsymbol \rho(\mathcal{U}_{\vec c})={\vec 0}\Leftrightarrow\mathcal{U}_{\vec c}=\{{\vec c}\}$. Thus, $\boldsymbol \rho={\vec 0}$ signifies zero privacy, where an adversary can exactly determine the associated input tuple.

Proposition~\ref{prop:privacy_index}-e) and f) provide a mechanism to compare different $\boldsymbol\rho$ vectors, wherefore the privacy of different solution methods can be compared. Proposition~\ref{prop:privacy_index}-e) and Remark~\ref{remark:knowledge} ensure that the optimization based approaches lacks robustness to adversarial attacks, such as those involving side information. In addition, Proposition~\ref{prop:privacy_index}-e) is useful to compute lower bounds for $\boldsymbol\rho$ of a complicated uncertainty set. The converse of Proposition~\ref{prop:privacy_index}-e) does not holds in general.

Proposition~\ref{prop:privacy_index}-f) is a weaker form of e): given a nonincreasing (with respect to $\R^3_+$) sequence $\big\{\boldsymbol\rho\big(\mathcal{U}^{(n)}_{{\vec c}}\big)\big\}_n^N$, $\exists$ a sequence $\big\{\bar{\mathcal{U}}^{(n)}_{{\vec c}}\big\}_n^N$ of sets, each containing ${\vec c}$ and $\mathcal{B}_{\bar{\mathcal{U}}^{(1)}_{{\vec c}}}\supseteq\cdots \supseteq\mathcal{B}_{\bar{\mathcal{U}}^{(N)}_{{\vec c}}}\supseteq\{{\vec c}\} $. Thus, Proposition~\ref{prop:privacy_index}-f) gives a way to compare the privacy of two algorithms given the privacy indexes, irrespective of the underlying solution machinery. In addition, Proposition~\ref{prop:privacy_index}-f) suggests a natural ordering of privacy indexes in general, (Figure~\ref{fig:PrivacyOrdering}) providing meanings to the \emph{smallest} (i.e., $\vec 0$) and the \emph{largest} (i.e., $(\infty,1,n)$) privacy for an \emph{input tuple} ${\vec c}\in\R^n$, \cf~Proposition~\ref{prop:privacy_index}-a).
}

\vspace{3mm}
\section{{Optimization Approaches and their Privacy}}\label{sec:Opt_Approaches_Privacy}

\addnew
{
In this section, standard optimization approaches are concisely presented, \cf~\texttt{Q2}~\S~\ref{subsec:problem_statement}. Examples are given to quantify the privacy of these optimization approaches, \cf~\texttt{Q1}. All the examples are designed specifically to illustrate our point.}

\vspace{-1mm}
\subsection{Transformation of Objective and Constraint Functions}\label{subsec:Obj-Cons_transformation}

\addnew
{
Suppose $\psi_0:\D_0\subseteq\R\rightarrow\R$ is monotonically increasing and $\D_0\supseteq\texttt{range}(f_0)$. Moreover, suppose $\psi_i:\D_i\subseteq\R\rightarrow\R$, with $\D_i\supseteq\texttt{range}(g_i)$, is such that $\psi_i(z)\leq 0\Leftrightarrow z\leq 0$ and $\varphi_i:\R\rightarrow\R$ is such that $\varphi({\vec z})= {\vec 0}\Leftrightarrow{\vec z}={\vec 0}$. Let \texttt{P} denote the problem formed by using  ${\psi}_0(f_0({\vec x}))$, ${\psi}_i(g_i({\vec x}))$, and ${\varphi}_i(h_i({\vec x}))$, instead of $f_0({\vec x})$, $g_i({\vec x})$, and $h_i({\vec x})$, respectively. Here $g_i$ and $h_i$ are $i$th component of $g$ and $h$, respectively. Then ${\vec x}^\star$ solves \texttt{P} if and only if it solves problem~\eqref{eq:optimization_prob}. With this method, potential obfuscation of problem data (\cf~Definition~\ref{def:Obsf_Unctnty}) is directly linked to the function compositions, ${\psi}_0(f_0(\cdot))$, ${\psi}_i(g_i(\cdot))$, and ${\varphi}_i(h_i(\cdot))$.
}

\addnew
{

\begin{figure}[t]
\centering
\begin{tikzpicture}[scale=0.45]
\draw[->] (0, 0.0272) -- coordinate (x axis mid) (7.5, 0.0272) node[right] {\tiny{$\texttt{a}(\mathcal{U}_c)$}};
\draw[->] (0, 0.0272) -- coordinate (y axis mid)(0,4.5) node[above left] {\tiny{$\nu(\mathcal{U}_c)$}};
\draw (0, 3.1564) -- (8.1, 3.1564) node[below] {\tiny{$\nu(\mathcal{U}_c)=0.8750$}};
\draw (2, 0.0272 ) -- (2, 4.5) node[above] {\tiny{$\texttt{a}(\mathcal{U}_c)=2$}};
\node at (6.9,0.75) {\scriptsize{$\mathcal{U}_c= \{c\in\R^{7} \left| \cdots\cdots \right.\} $}};

\draw[color=red] plot[only marks,mark=ball, ball color=red] file {abcd1.table};
\draw[color=red] plot[only marks,mark=ball, ball color=blue] file {abcd2.table};
\draw[color=red] plot[only marks,mark=ball, ball color=red] file {abcd3.table};
\draw[color=red] plot[only marks,mark=ball, ball color=blue] file {abcd4.table};
\draw[color=red] plot[only marks,mark=ball, ball color=red] file {abcd5.table};
\draw[color=red] plot[only marks,mark=ball, ball color=blue] file {abcd6.table};
\draw[color=red] plot[only marks,mark=ball, ball color=red] file {abcd7.table};
\draw[color=red] plot[only marks,mark=ball, ball color=blue] file {abcd8.table};
\draw[color=red] plot[only marks,mark=ball, ball color=red] file {abcd9.table};
\draw[color=red] plot[only marks,mark=ball, ball color=blue] file {abcd10.table};

 \foreach \x in {0,1,2,3,4,5,6,7}
        \draw [xshift=-0](\x,0.0272) -- (\x ,-0.05)
            node[anchor=north] {\tiny{$\x$}};
\foreach \y/\ytext in {   0.0272/0,  0.1515/0.5, 0.8838/0.75, 4.2396/0.9}
        \draw (1pt,\y ) -- (-3pt,\y cm) node[anchor=east] {\tiny{$\ytext$}};
\draw [fill=blue,opacity=0.07] (2,3.1564) rectangle (7,4.5);
\draw [fill=red,opacity=0.07] (0,0.0272) rectangle (2,3.1564);
\draw [-to,shorten >=2pt, thick](7, 4.7) -- (6, 4.2396) ;
\draw [-to,shorten >=2pt, thick](6, 4.2396) -- (3, 3.1564) ;
\draw [-to,shorten >=2pt, thick] (3, 3.1564)-- (1, 1.4160) ;
\node at (16, 4.2396) {\footnotesize{arrows $\Rightarrow$ improvements in $\mathcal{K}$}};
\end{tikzpicture}
\captionof{figure}{\label{fig:PrivacyOrdering}\small{\addnew{Feasible $\boldsymbol\rho(\mathcal{U}_c)$ for a fixed $\texttt{d}(\mathcal{U}_c)$. The privacy corresponds to the points in the upper-right shaded area are better than that of the points in lower-left shaded area. Rest of the points are incomparable.}}}
\vspace{-6mm}
\end{figure}
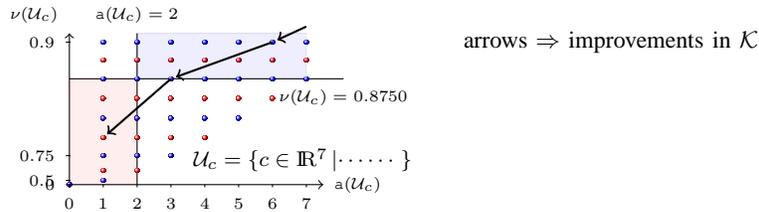
}

\addnew{Scaling used in~\cite{Dreier-Kerschbaum-2011,Bednarz-Bean-Roughan-2009,Wang-Ren-Wang-11,Bednarz-2012} is a particular case of the objective and constrained transformation. In~\cite{Mangasarian-OptLet-2012}, authors handle their affine constraint functions, denoted ${\vec A}{\vec x}{-}{\vec b}{=}{\vec 0}$, by using $\varphi(\vec z){=}{\vec B}{\vec z}$, where ${\vec B}$ is full column rank with the desired properties ${\vec A}{\vec x}{-}{\vec b}{=}{\vec 0}{\Leftrightarrow}{\vec B}({\vec A}{\vec x}{-}{\vec b}){=}{\vec 0}$. Now, we present a new, yet an important example to convey the idea.}
\begin{Exx}[\addnew{Localization from range measurements}]\label{ex:L_norm_LS_problem}
\addnew
{
Suppose the area $\mathcal{A}\in\R^2$ with $\texttt{d}(\mathcal{A})=D$ contains $N$ military deployed beacons, denoted $i=1,\ldots,N$, that serve civilian boats to locate their positions. Let $\bar{\vec x}\in\R^2$ denote the position of boat \texttt{B}. The range measurements, denoted $y_i\in\R$ and ${\vec a}_i\in\R^2$, $i=1,\ldots,N$, from the beacons are received by \texttt{B} to compute its location ${\vec x}$. In particular, $y_i$ and ${\vec a}_i$ are the distance and the unit vector, respectively, from \texttt{B} to beacon $i$. Assuming a linear measurement model with unknown noise, an estimate $\hat{\vec x}$ of ${\vec x}$ is given by the solution of the problem, $\mbox{minimize}  ||{\vec A}{\vec x}-{\vec y}||_2$,
where the variable is ${\vec x}\in\R^2$, and the problem data are ${\vec A}=-[{\vec a}_1,\ldots,{\vec a}_N]\tran$ and ${\vec y}=(y_1,\ldots,y_N)$. However, due to security reasons, military beacons do not want to send their range measurements, which are simply the polar coordinates of their positions, with respect to \texttt{B}.
%
%
Now the idea is to simply use the objective transformation $\psi_0({z})= z^2$, and let \texttt{B} solve
\be\label{eq:obj_squaring}\nonumber
\mbox{minimize}_{{\vec x}} \ {\vec x}\tran(\textstyle\sum_{i=1}^N{\vec a}_i{\vec a}_i\tran){\vec x}+ 2(\textstyle\sum_{i=1}^N{y}_i{\vec a}_i\tran){\vec x}\triangleq{\vec x}\tran\bar{\vec A}{\vec x}+2 {\bar{\vec y}}\tran{\vec x} \ .
\ee
The beacons communicate with each other to compute $\bar{\vec A}$ and $\bar{\vec y}$, which is transmitted to \texttt{B}, i.e., the message $S{=}(\bar{\vec A},\bar{\vec y})$. Let us now explore the privacy of this method on the input tuple ${\vec c}{=}y_1{\vec a}_1$, the location of beacon~$1$. Without loss of generality, let $\bar{\vec x}{=}(0,0)$. Suppose \texttt{B} knows by definition $||{\vec a}_1||{=}1$, $\bar{\vec A} {=}{\vec a}_1{\vec a}_1\tran{+}\sum_{i=2}^{N}({\vec a}_i{\vec a}_i\tran)$, 3) $\bar{\vec y}{=}y_1{\vec a}_1 {+}\sum_{i=2}^{N}y_i{\vec a}_i$, $y_i{>}0$, and, of course, $N{=}N_0$, i.e., the auxiliary knowledge.
\addnew
{\vspace{-0mm}
\begin{figure}[h!]
\centering
\subfigure[]{
\begin{tikzpicture}[scale=0.35]
\fill [color=blue!10,draw=black,rotate=45] (0,0) ellipse (1.7cm and 1.2cm);
\fill [color=blue!10,rotate=135,opacity=0.0] (0,0) ellipse (4cm and 0.01cm);
\draw (0,0) circle (1cm);
\draw [thin] (0, 2) -- (0, -2);
\draw [thin] (-3, 0) -- (3, 0);
\fill [draw=black] (-1.4142, 1.4142) circle (2.0pt) node [above]{\tiny{$\bar{\vec y}$}};
\fill [draw=black] (0,0) circle (2.0pt) node [above right ]{\tiny{$\bar{\vec x}$}};
\end{tikzpicture}
\label{fig:MinmaxVsSum-b}}
\subfigure[]{
\begin{tikzpicture}[scale=0.35]
\fill [color=blue!10,draw=black,rotate=45] (0,0) ellipse (2.1cm and 0.55cm);
\fill [color=blue!10,rotate=135,opacity=0.0] (0,0) ellipse (4cm and 0.01cm);
\draw (0,0) circle (1cm);
\draw [thin] (0, 2) -- (0, -2);
\draw [thin] (-3, 0) -- (3, 0);
\fill [draw=black] (-1.4142, 1.4142) circle (2.0pt) node [above]{\tiny{$\bar{\vec y}$}};
\draw [draw=black] ( 0.2419   , 0.9703) circle (2.2pt) ; 
\draw [draw=black] ( -0.2419   , -0.9703) circle (2.2pt) ;
\draw [draw=black] ( -0.9703,-0.2419) circle (2.2pt) ;
\draw [draw=black] ( 0.9703,0.2419) circle (2.2pt) ; 
\draw (0,0) edge [-latex, thin] (  0.4697,  1.8839) ;
\draw (0,0) edge [-latex, thin] ( -1.8839 ,-0.4697  ) ;
\draw (-1.4142,1.4142) edge [-, thin, opacity=0.2] (  0.4697,  1.8839) ;
\draw (-1.4142,1.4142) edge [-, thin, opacity=0.2] (   -1.8839 ,-0.4697 ) ;
\node (C) at (0.4697,  2) {\scriptsize{${\vec c}_1$}};
\node (C) at (-2 ,-0.7) {\scriptsize{${\vec c}_2$}};
\end{tikzpicture}
\label{fig:MinmaxVsSum-a}}
\subfigure[]{
\begin{tikzpicture}[scale=0.35]
\draw [thick] (-2.5, 2.5) -- (2.5, -2.5);
\draw (0,0) circle (1cm);
\draw [thin] (0, 2) -- (0, -2);
\draw [thin] (-3, 0) -- (3, 0);
\fill [draw=black] (-1.4142, 1.4142) circle (2.0pt) node [above]{\tiny{$\bar{\vec y}$}};
\fill [draw=black] (0,0) circle (2.0pt) node [above right ]{\tiny{$\bar{\vec x}$}};
\node (C) at (-2.9, 2.5) {\scriptsize{${A}_1$}};
\node (C) at (2.9, -2.5) {\scriptsize{${A}_2$}};
\end{tikzpicture}
\label{fig:MinmaxVsSum-a}}
\vspace{-5mm}
\end{figure}
}
Thus, $\mathcal{U}_{{\vec c}}=\{{\vec c}=y_1{\vec a}_1 \ | \ ||{\vec a}_1||=1, {\vec a}_i{\in}{\mathcal{E}}(\bar{\vec A}), y_1{\vec a}_1{=} (\bar{\vec y}{-}\sum_{i=2}^{N_0}y_i{\vec a}_i), ||{\vec a}_i||{=}1 \ \forall i{=}2,\ldots,N_0, y_i{\leq} D \ \forall i{=}1,\ldots,N_0\}$, where $\mathcal{E}(\bar{\vec A})=\{{\vec a} \ | \  {\vec a}\tran{\bar{\vec A}}^{-1}{\vec a}{\leq} 1\}$.  Note that ${\vec a}_i{\in}{\mathcal{E}}(\bar{\vec A}) \ \forall i$ is a necessary condition, by definition of $\bar{\vec A}$. Let us next compute~$\boldsymbol \rho(\mathcal{U}_{\vec c})$.


We first build some intuition of $\mathcal{U}_{{\vec c}}$. The figure above illustrates all possible cases, the circle and the shaded area represent the unit circle $\mathcal{O}$ and the ellipse $\mathcal{E}$, respectively. So, ${\vec a}_1\in\mathcal{O}\cap\mathcal{E}$. Note that $y_i$ is simply the magnitude of some vector in $\R^2$, which is spanned by vectors in $\mathcal{O}\cap \mathcal{E}$. Now consider scenario 1, where $N_0>2$. The scenario~1-cases (a),(b) yield ${\boldsymbol \rho}_{\textrm{1,a,b}}(\mathcal{U}_{{\vec c}})=(D,1,2)$, where the subscript of ${\boldsymbol \rho}$ denotes the scenario and the cases. The scenario 1-case~(c) yields ${\boldsymbol \rho}_{\textrm{1,c}}(\mathcal{U}_{{\vec c}})=(D,1,1)$. Now consider scenario 2, where $N_0=2$. Clearly, the scenario~2-case~(a) doesn't exist. If scenario~2-case~(b), $\mathcal{O}\cap \mathcal{E}$ is a set with $4$ points, i.e., ${\vec a}_1$ is on one of the $4$ points, \cf~the open circles in (b). Moreover, because \texttt{B} knows $\bar{\vec y}$, ${\vec c}=y_1{\vec a}_1$ is either at ${\vec c}_1$ or ${\vec c}_2$, i.e., $\mathcal{U}_{\vec c}=\{{\vec c}_1,{\vec c}_2\}$, \cf~figure (b). Thus, ${\boldsymbol \rho}_{\textrm{2,b}}(\mathcal{U}_{\vec c})=(||{\vec c}_1-{\vec c}_2||,0.5,1)$. If scenario~2-case~(c), we similarly get $\mathcal{U}_{\vec c}=A_1A_2$, the line segment of length $D$, \cf~figure (c), i.e., ${\boldsymbol \rho}_{\textrm{2,c}}(\mathcal{U}_{{\vec c}})=(D,1,1)$. The ${\boldsymbol \rho}$ when $N_0=1$ is similarly computed.

Based on our privacy index, ${\boldsymbol \rho}_{\textrm{1,a,b}}\geq {\boldsymbol \rho}_{\textrm{1,c}}={\boldsymbol \rho}_{\textrm{2,c}}\geq {\boldsymbol \rho}_{\textrm{2,b}}$. The first inequality clearly agrees with out intuition. The second equality agrees with our intuition in the sense that the ambiguity in $c$ is identical in the associated two settings. The last inequality can be intuitively explained as follows: scenario 1,2-case (c) corresponds to a setting where the beacon measurements are linearly dependent, and therefore the beacon~1 can be any point along the line $A_1A_2$, as opposed to scenario 2-case (b), where the beacon measurements are not linearly dependent.}
\end{Exx}

\subsection{Transformation via Change of Variables}\label{subsec:chg_of_variable}

\addnew
{
Suppose the decision variable ${\vec x}$ in problem~\eqref{eq:optimization_prob} is in $\R^n$. Let $\phi:\R^m\rightarrow \R^n$ be a function  such that $\texttt{range}(\phi)\supseteq \bigcap_{i} \texttt{dom}(g_i)$. Now consider the problem, denoted \texttt{Q}, achieved by the change of variables ${\vec x}= \phi ({\vec z})$. Then \texttt{Q} is equivalent to problem~\eqref{eq:optimization_prob} in the sense that if ${\vec z}^\star$ solves \texttt{Q}, then ${\vec x}^\star= \phi({\vec z}^\star)$ solves problem~(\ref{eq:optimization_prob}). Note again that any potential obfuscation of problem data (\cf~Definition~\ref{def:Obsf_Unctnty}) is directly linked to the function compositions, $f_0(\phi(\cdot))$, $g(\phi(\cdot))$, and $h(\phi(\cdot))$.
}

\addnew
{
All the approaches in~\cite{Mangasarian-OptLet-2011,Dreier-Kerschbaum-2011,Mangasarian-OptMethSW-2011,Bednarz-Bean-Roughan-2009,Wang-Ren-Wang-11,Bednarz-2012,Mangasarian-Wild-2008,Mangasarian-Wild-Fung-2008} can be obtained as a particular case of the \emph{change of variables}. In particular, those methods uses an affine $\phi$. Let us now give an example for a case where $\phi$ is not necessarily affine.}
\begin{Exx}[\addnew{Resource allocation}]\label{ex:nonlinear_change_variables}
\addnew
{
Suppose an ammunition supply depot (\texttt{ASD}) wants to transfer $N$-types of nuclear-ammunition to the war field. Let $\beta_i$ denote the radiation per unit volume of $i$th ammunition type. The overall radiation should be no greater than a specified value $\gamma$. The containers for each type~$i$ should be cubic of length $x_i$. The criterion for computing the volume of type~$i$ is determine by the penalty function $\alpha_i/x_i$, where $\alpha_i$ denote the priority levels of the $i$th type. Suppose the \texttt{ASD} relies on a third party to compute the desired container box lengths ${\vec x}^\star=(x^\star_i)_{i=1:N}$, by $\mbox{minimizing}~\sum_{i=1}^{N} (\alpha_i/x_i)$ $\mbox{subject to}~\sum_{i=1}^{n} \beta_i x^3_i \leq \gamma$ and ${\vec x} \geq {\vec 0}$. The problem data is $\alpha_i>0$, $\beta_i>0$, and $\gamma>0$. On the other hand, suppose the ASD does not want to reveal the input tuples $(\alpha_i,\beta_i)$, $i=1,\ldots,N$.

What if ASD uses the change of variable $x_i = \phi_i(z_i)= \alpha_i\exp(z_i)$ for all $i$. Then it is sufficient to ask the third party (say \texttt{B}) to $\mbox{minimize}~\sum_{i=1}^{N} \exp(-z_i)$ $\mbox{subject to}~\sum_{i=1}^{N} s_i \exp(3z_i) \leq 1$,
with variable ${\vec z}=(z_1,\ldots,z_N)$ and problem parameters $s_i=(\beta_i\alpha^3_i/\gamma)$.
This method requires the transmission of the message $S=(s_i)_{i=1:N}$ to \texttt{B}. Suppose \texttt{B} knows the dependence of $s_i$ on $\alpha_i,\beta_i,\gamma$ and the positivity of $\alpha_i,\beta_i,\gamma$. Therefore, $\mathcal{U}_{(\alpha_i,\beta_i)}{=}\{(\alpha_i,\beta_i) \ | \ s_i{=}\beta_i\alpha^3_i/\gamma , \beta_i,\alpha_i,\gamma>0  \}$, which corresponds to $\boldsymbol \rho(\mathcal{U}_{(\alpha_i,\beta_i)})=(\infty,1,2)$. This is the highest possible privacy for an input tuple in $\R^2$, \cf~Proposition~\ref{prop:privacy_index}-(a).}
\end{Exx}

One can readily apply hybrid variants of the \emph{transformation via change of variables} and \emph{transformation of objective and constraints},~\cf~\cite{Dreier-Kerschbaum-2011,Wang-Ren-Wang-11}.

\subsection{Decomposition Methods}\label{subsec:decomposition_based_method}
\addnew
{
In this section, we highlight important aspects of decomposition methods for preserving privacy of problem data. In particular, the dual decomposition method is considered. Arguments, in the case of primal decomposition methods and ADMM method are similar.}


\addnew
{
We start by noting that the \emph{decomposition} methods are applied whenever the problem structure inherently possesses some separability properties, where the original problem can be solved by coordinating smaller subproblems, one for each subsystem. Thus, instead of problem~\eqref{eq:optimization_prob}, the general problem we consider in this section is given by
\begin{equation} \label{Primal_problem_big}
\begin{array}{ll}
\mbox{minimize} & \sum_{i=1}^{N} f_i ({\vec x}_i, {\vec y}_i;{\vec c}_i) \\
\mbox{subject to} & ({\vec x}_i, {\vec y}_i) \in \mathcal{G}_i({\vec c}_i),  \qquad   i=1,...,N\\
& {\vec y}_i={\vec E}_i {\vec z}  , \qquad  \quad  \  i=1,...,N \ ,
\end{array}
\end{equation}
where the variables are $({\vec x}_i,{\vec y}_i)_{i=1:N}$ and ${\vec z}$, and $i$ denotes the subsystem index. Typically, ${\vec x}_i$, ${\vec y}_i$ and ${\vec z}$ are called local, interface, and global variables, respectively. The matrix ${\vec E}_i$ is a $0$-$1$ matrix that maps ${\vec z}$ into the interface variables ${\vec y}_i$. This mapping accounts for the coupling of the problem and ensures consensus between subsystem interface variables and the global variables. The set $\mathcal{G}_i({\vec c}_i)$ is the constraints at the subsystem~$i$ and is dependent on the private problem data~${\vec c}_i$. }
\addnew{
Let us now summarize the standard dual decomposition algorithm and see how it preserves privacy for private data~${\vec c}_i$.
\begin{Algg}\label{alg:dual_decomp_algorithm}
\noindent \emph{Dual Decomposition}

Given ${\boldsymbol\lambda}=(\boldsymbol{\lambda}_1,\ldots,\boldsymbol{\lambda}_N)$ such that ${\vec E}\tran\boldsymbol{\lambda} = {\vec 0}$. Set $k=1$.

\textbf{while} \mbox{(stopping criterion)}

\begin{enumerate}
\item[1.] Subsystem~$i$ $\mbox{minimizes}~f_i ({\vec x}_i,{\vec y}_i;{\vec c}_i) +\boldsymbol{\lambda}_i\tran {\vec y}_i$ $\mbox{subject to}~({\vec x}_i,{\vec y}_i) \in \mathcal{G}_i({\vec c}_i)$, with $({\vec x}_i,{\vec y}_i)$. Let $\big({\vec x}^{(k)}_i({\vec c}_i),{\vec y}^{(k)}_i({\vec c}_i)\big)$ denote the solution.
\item[2.] Subsystems coordinate to update global ${\vec z}$, i.e., ${\vec z}\asign({\vec E}\tran{\vec E})^{-1}{\vec E}\tran \big({\vec y}^{(k)}_1({\vec c}_1),\ldots,{\vec y}^{(k)}_N({\vec c}_N)\big)$.
\item[3.] Subsystem~$i$ updates ${\boldsymbol\lambda}_i$ as, ${\boldsymbol\lambda}_i \asign {\boldsymbol\lambda}_i+\alpha_k\big({{\vec y}^{(k)}_i({\vec c}_i)}-{\vec E}_i{\vec z}\big)$. Set $k\asign k+1$.
\end{enumerate}
\end{Algg}
%
%
%

In the algorithm above, $\alpha_k$ is an appropriate step size~\cite{Bertsekas-Tsitsiklis-97}. As far as the privacy of the problem data is concerned, the crucial step of the algorithm is step~2. This is because the coordination requires exchange of messages, which are dependent on the problem data. In particular, the $i$th subsystem requires revealing message, denoted $S^{\textrm{dual}}_i$, where $S^{\textrm{dual}}_i=\big({\vec y}^{(k)}_i({\vec c}_i)\big)_{k=1:T}$ and $T$ is the number of iteration before the algorithm termination. Note that ${\vec y}^{(k)}_i({\vec c}_i)$ is, in fact, the subgradients of dual master objective function~\cite{Bertsekas-Tsitsiklis-97}. Therefore, any potential obfuscation of problem data of subsystem~$i$ is directly linked to the computed subgradients until the algorithm termination, i.e., $S^{\textrm{dual}}_i$. An example is given below.
}
\addnew{
\begin{Exx}[Average consensus]\label{ex:QP_multi_party_constraints_coupled}
Consider the computation of the arithmetic average $\bar c$ of $N>2$ scalars, denoted $c_1,\ldots,c_N$, where each $c_i$ is private to subsystem~$i$. Note that $\bar c$ is simply the argument $y$ that minimizes $\sum_{i=1}^{N}(y-c_i)^2$. This minimization can be equivalently posed as
\begin{equation} \label{eq:AVG-Consensus-optimization_prob}
\begin{array}{ll}
\mbox{minimize} & \sum_{i=1}^{N}(y_i-c_i)^2 \\
\mbox{subject to} & y_i=z, \ i=1,\ldots,N \ ,
\end{array}
\end{equation}
with variables $(y_i)_{i=1:N}$ and $z$. Let us next show, without loss of generality, how Algorithm~\ref{alg:dual_decomp_algorithm} achieves privacy for input tuple $c_1$ of subsystem~$1$, given adversary is $N$.

Note that the step~1 of Algorithm~1 is
\vspace{-5mm}
\be\label{eq:consensus}
y^{(k)}_i(c_i)\asign (c_i-\lambda_i/2) \ .
\ee
The step~2 is simply the average $z\asign(1/N)\sum_{i=1}^N (c_i-\lambda_i/2)$ and step~3 is ${\lambda}_i\asign{\lambda}_i+\alpha_k\big({{ y}^{(k)}_i({ c}_i)}-{z}\big)$. Thus, the message exposed by subsystem~$1$ to subsystem~$N$ is $S_1=\big(y^{(k)}_1({c}_1)\big)_{k=1:T}$, where $T$ is the total iterations.
%

Without explicitly describing the relation $\mathcal{R}_{{c}_1}$ (or equivalently the uncertainty set $\mathcal{U}_{{c}_1}$), we use an alternative method, where the associated $\boldsymbol \rho (\mathcal{U}_{{ c}_1})$ is computed by construction. Now note that if for $i=1,\ldots,N-1$, ${c}_i$ and $\lambda_i$ satisfy \eqref{eq:consensus} for all $k=1,\ldots,T$, so is $\hat{{c}}_i(\theta_i)=(c_i-\theta_i,\lambda_i-2\theta_i)$, where $\boldsymbol \theta=(\theta_i)_{i=1:N-1}$ is chosen such that ${\vec 1}\tran \boldsymbol\theta = 0$. Indeed, $\theta_1$ here parameterizes the uncertainty set $\mathcal{U}_{c_1}$. In particular, $\mathcal{U}_{c_1}=\{c_1-\theta_1 \ | \ {\vec 1}\tran \boldsymbol\theta = 0 \}$, which yields $\boldsymbol \rho(\mathcal{U}_{c_1})=(\infty,1,1)$. This is the highest possible privacy for an input tuple in $\R$, \cf~Proposition~\ref{prop:privacy_index}-a).

Even though, we assume that subsystem $N$ is a passive adversary (\cf~Definition~\ref{def:adversary}), it is easily seen that, even if $r$ number of subsystems collude, $\boldsymbol \rho(\mathcal{U}_{c_1})=(\infty,1,1)$ remains intact, given $r<N-1$. In other words, as long as subsystem~$1$ has at least one reliable subsystem~$j$, $j\neq 1$, $\boldsymbol \rho(\mathcal{U}_{c_1})=(\infty,1,1)$.
\end{Exx}
}
\addnew
{
}

\section{Conclusions}\label{sec:conclusions}
\addnew{
In this paper, a mechanism to quantify the privacy of a broader class of optimization approaches was proposed. We defined a one-to-many relation, which relates a given adversarial observed message to an uncertainty set of the problem data. The privacy measure was then designed based on the uncertainty sets. The properties of the proposed privacy measure was analyzed. From the definition it is clear that optimization based approaches lacks robustness to adversarial attacks, such as those involving side information. Our privacy definition clearly model those intrinsic properties of optimization approaches. The key ideas were illustrated by examples. An important possible extensions is to explore ways of designing optimization approaches to guarantee a specified privacy level.

}

\bibliographystyle{IEEEbib}
\bibliography{jour_short,conf_short,references,references_PPO}
\end{document}